\documentclass{amsart}
\usepackage{graphicx}
\newtheorem{thm}{Theorem}[section]

\theoremstyle{definition}

\newtheorem{fact}[thm]{Fact}
\theoremstyle{remark}

\numberwithin{equation}{section}
\newcommand{\eps}{\varepsilon}

\newcommand{\yl}{{\check{y}}}
\newcommand{\yh}{{\hat{y}}}
\newcommand{\bl}{{\check{b}}}
\newcommand{\bh}{{\hat{b}}}

\def\id{\operatorname{Id}}
\def\supp{\operatorname{supp}}

\def\R{\mathbb R}
\def\eps{\varepsilon}
\def\polylog{\operatorname{polylog}}

\title{Almost Optimal Unrestricted Fast Johnson-Lindenstrauss Transform}
\author{Nir Ailon}
\address{Technion, Haifa, Israel}
\email{nailon@gmail.com}
\thanks{Nir Ailon's affiliation: Technion, Haifa, Israel, nailon@gmail.com}
\author{Edo Liberty}
\address{Yahoo! Research, Haifa, Israel}
\email{edo@yahoo-inc.com}
\thanks{Edo Liberty's affiliation: Yahoo! Research, Haifa, Israel edo@yahoo-inc.com}
\date{May 30, 2010}

\begin{document}

\maketitle

% ----------------------------------------------------------------
\begin{abstract}
The problems of random projections and sparse reconstruction
have much in common and individually received much attention.
Surprisingly, until now they progressed in parallel and remained mostly separate.
Here, we employ new tools from probability in Banach spaces that were successfully used in the context of
sparse reconstruction to advance on an open problem in random pojection.
In particular, we generalize and use an intricate result by Rudelson and Vershynin for sparse reconstruction which uses
Dudley's theorem for bounding Gaussian processes.
%we use techniques developed by Rudelson and Vershynin for compressed
%sensing to advance on an open problem in computational dimension reduction.
%This paper takes this well researched problem another step forward.
%It employs new tools from probability in Banach space that were successfully used in the context of
%sparse reconstruction by Rudelson and Vershynin, in particular, Dudley's theorem for bounding Gaussian processes.
%More precisely, we show how to linearly transform  any set of $N$ real vectors
%in $n$ dimensional space into a space of dimension $k=O(\log N\polylog(n))$, while (1) preserving the pairwise
%distances among the vectors to within any constant distortion and (2) being able to apply the transformation
%in time $O(n\log n)$ on each vector.
Our main result states that any set of $N = \exp(\tilde{O}(n))$ real vectors
in $n$ dimensional space can be linearly mapped to a space of dimension $k=O(\log N\polylog(n))$, while
(1) preserving the pairwise distances among the vectors to within any constant distortion and
(2) being able to apply the transformation in time $O(n\log n)$ on each vector.
This improves on the best known $N = \exp(\tilde{O}(n^{1/2}))$
achieved by Ailon and Liberty and  $N = \exp(\tilde{O}(n^{1/3}))$ by Ailon and Chazelle.
 The dependence in the distortion constant however is believed to be suboptimal and subject to further investigation.
For constant distortion, this settles the open question posed by these authors up to a $\polylog(n)$ factor
while considerably simplifying their constructions.
%and solves, up to polylogarithmic factors, the open question posed by them.
%\footnote{The notation $\tilde{O}(\cdot)$ suppresses arbitrarily small polynomial coefficients and polylogarithmic factors.}
%Until now, simultaneously meeting both properties required that  $k$ be $O(n^{1/2-\epsilon})$ (for arbitrarily small $\epsilon$).

%Up to the $\polylog(n)$ factor our result settles the open question posed by Ailon and Liberty and Ailon and Chazelle.
%The dependence in the distortion constant however is believed to be suboptimal and subject to further investigation.
%Until now all dimension reduction analyses relied on a simple union bound on
%certain failure probabilities. Rudelson's and Vershynin's result shows how to bypass
%this union bound by using a theorem due to Dudley on Gaussian processes.
\end{abstract}

% ----------------------------------------------------------------
\section{Introduction}
Designing computationally efficient transformations that reduce dimensionality of data while approximately
preserving its metric information lies at the heart of many problems.
While in compressed sensing such techniques are sought for sparse data in a real or complex metric space (with
respect to some basis), in random projections, following the seminal work of Johnson and Lindenstrauss, one seeks to reduce dimension of any
set of finite data.\footnote{The term "random projections" describes  Johnson and Lindenstrauss'
s original construction
and became synonymous with the process of approximate metric preserving dimension reduction using randomized linear mappings.
However, these linear mappings need not be (and indeed are usually not) projections in the linear algebraic sense of the word.}
In both applications, random matrices of a suitable size
\cite{JL84}\cite{FranklM87}\cite{DasGuptaGupta99}\cite{Achlioptas03} result in optimal
construction \cite{Alon03} in the parameters $n$ (the original dimension), $k$ (the target dimension), $N$ (the number
of input vectors) and $\delta$ (the distortion).
However, these constructions' resulting running time complexity, measured as number of operations
needed in order to  map a vector, is suboptimal.

A major open question is that of designing such matrix distributions that can be applied efficiently to any vector, with optimal
dependence in the parameters $n,k,N$ and $\delta$.  Applications for such transformations were found e.g. in designing fast
approximation algorithms for solving large scale linear algebraic operations (e.g. \cite{Sarlos06}, \cite{Woolfe08})
The two lines of work, though sharing much in common, have mostly progressed in parallel.  Here we combine recent work
on bounds for sparse reconstruction to improve bounds of Ailon and Chazelle \cite{AilonCh06,AilonC10} and Ailon and Liberty Liberty \cite{AilonLi08} on
fast random projections, also known as Fast Johnson-Lindenstrauss transformations.
The new bounds allow obtaining the well known Fast Johnson-Lindenstrauss Transform for  finite sets of bounded cardinality
$N = \exp(\tilde{O}(n))$ where $n$ is the original dimension.
%The best known so far was obtained by Ailon and Liberty for sets of size up to $N=\exp\{O(n^{1/2-\epsilon})\}$ for arbitrarily small $\epsilon$
 The best known so far was obtained by Ailon and Liberty for sets of size up to $N=\exp\{\tilde{O}(n^{1/2})\}$.\footnote{The
 notation $\tilde{O}(\cdot)$ suppresses arbitrarily small polynomial coefficients and polylogarithmic factors.}
The latter improved on Ailon and Chazelle's original bound of $N=\exp\{O(n^{1/3})\}$,
which initiated the construction of Fast Johnson-Lindenstrauss Transforms.
We also mention Dasgupta et al.'s work \cite{DasguptaKS10} on construction of Johnson-Linenstrauss random matrices which can be more efficiently
applied to sparse vectors, with applications in the streaming model, and Ailon et al's work \cite{DBLP:conf/approx/LibertyAS08} on design of Johnson-Lindenstrauss matrices
that run in linear time under certain assumptions on various norms of the input vectors.

The transformation we derive here is a composition of two random matrices:  A random sign matrix and a random
selection of a suitable number $k$ of rows from a Fourier matrix, where $k=O(\delta^{-4}(\log N) \polylog(n))$, and
$\delta$ is the tolerated distortion level.  The result, for constant $\delta$, is believed to be suboptimal within the $\polylog(n)$ factor
in the target dimension $k$.  The running time of performing the transformation on a vector is
dominated by the $O(n\log n)$ of the Fast Fourier Transform, and is believed to be optimal.  The possibility of
obtaining such a running time for fixed distortion was left as an open problem in Ailon and Chazelle and Ailon
and Liberty's work, and here we resolve it up to a factor of $\polylog(n)$.  The dependence on the constant $\delta$ is also believed
to be suboptimal, and the ``correct'' dependence shoould be $\delta^{-2}$.  The question of improving this dependence
is left as an open problem.
%The questions of removing the $\polylog n$ dependence and improving the dependence on $\eps$ remain challenging open problems.

The use of a combination of random sign matrices and various forms of subsampled Fourier matrices was also used in the
work of Ailon and Chazelle \cite{AilonCh06} and later Ailon and Liberty \cite{AilonLi08}, as well as that of Matousek \cite{Matousek06}.
Here we obtain improved analysis using recent work by Rudelson and Vershynin for sparse reconstruction \cite{rudelsonVershynin}.
%Unlike earlier work on computational dimension reduction, the result
%is not obtained by union bounding the probability of failure events over all pairs of vectors.

\subsection{Restricted Isometry}

An underlying idea  common to both random projections and sparse reconstruction
is the preservation of metric information under a dimension reducing transformation.
In sparse reconstruction theory, this property is known as \emph{restricted isometry}
 \cite{CandesRT06}\cite{Donoho06}.
A matrix $\Phi$ is a restricted isometry with sparseness paramater $r$ if for some $\delta>0$,
\begin{equation}
\forall \;\; r\mbox{-sparse} \; y\in\R^n \;\;\;(1-\delta)\|y\|_2^2 \leq \|\Phi y\|_2^2 \leq
(1+\delta)\|y\|_2^2\ .
\end{equation} %where $C$ is some global nonnegative constant,

By $r$-sparse $y$ we mean vectors in $\R^n$ with all but at most $r$ coordinates zero.
It was shown in  \cite{CandesRT06} that
the restricted isometry property  is sufficient for the
purpose of perfect reconstruction of sparse vectors, \emph{compressed sensing} being one of the
prominent applications. %In \cite{rudelsonVershynin} Rudelson and Vershynin
%prove a tighter version of \cite{CandesT06},

% ----------------------------------------------------------------
%\section{Rudelson and Vershynin's sparse RIP result}
In \cite{Rudelson06sparsereconstruction}, Rudelson and Vershynin construct a distribution over
$k\times n$ matrices $\Phi$ such that,
with high probability, $\Phi$ has the restricted isometry property with sparseness parameter $r$
and arbitrarily small $\delta>0$.\footnote{Their analysis is done over the complex
field, but we restrict the discussion to the reals here.}
%By $r$-sparse we mean that all but at most $r$ coordinates of $y$ are zero.
In their analysis, $k =
O(\delta^{-2}r\log(n)\cdot \log^2(r)\log(r\log n))$ and $\Phi$ can be applied
(to a given vector $x$) in running time $O(n\log n)$.  Assuming $r$ polynomial in $n$, this takes the simpler
form of  $k=O(\delta^{-2}r\log^4 n)$.\footnote{In their work, the dependence of $k$ on $\delta$
is not analyzed because $\delta$ is assumed to be fixed (for sparse signal reconstruction purposes, this dependence is not
important).  It is not hard to derive the quadratic dependence of $k$ in $\delta^{-1}$ from their work.} In fact, $\Phi$
is (up to a constant) nothing other than a random choice of $k$ rows from the (unnormalized)
 Hadamard matrix, defined as $\Psi_{\omega, t} = (-1)^{\langle \omega, t \rangle}$, where $\langle \cdot, \cdot \rangle$
is the dot product over the binary field, $n$ is assumed to be a power of $2$
and $\omega, t$ are thought of as $\log n$ dimensional vectors over the binary
field in an obvious way.\footnote{Rudelson and Vershynin use the complex Discrete Fourier Transform matrix, but
their analysis does not change when using the Hadamard matrix.}
  As a corollary of the result, one obtains a universal
matrix for reconstructing sparse signals, which can be applied to a vector in time $O(n\log n)$.
The conjecture is that
the same distribution with $k=O(\delta^{-2}r\log n)$ should work as well, but
this is a major open question  beyond the scope of this work.  For an excellent survey explaining how restricted
isometry can be used for sparse reconstruction, and why designing such matrices with good computational properties
is important we refer the readers to  \cite{BrucksteinDE09} and to references therein.

Independently, Ailon and Chazelle \cite{AilonCh06} and Ailon and Liberty
\cite{AilonLi08} were interested in constructing a distribution of $k\times n$
matrices $\Phi$ such that for any set $Y\subseteq \R^n$ of
cardinality $N$, one gets %
\begin{equation}\label{JL}
\forall \;\; y\in Y \;\;(1-\delta)\|y\|_2^2 \leq \|\Phi y\|_2^2 \leq (1+\delta)\|y\|_2^2,
\end{equation}
with constant probability.  Additionally, the number of steps
required for applying $\Phi$ on any given $x$ is $O(n\log n )$. In
their result $k$ was taken as $O(\delta^{-2}\log N)$, which is also
essentially the best possible \cite{Alon03}. Unfortunately, both
results break down when $k = \Omega(n^{1/2})$.\footnote{Ailon and Chazelle \cite{AilonCh06} and Ailon and Liberty \cite{AilonLi08}
used $d$ to denote the data dimension, $n$ its cardinality and $\eps$ the sought distortion bound.
Here we follow Rudelson and Vershynin's convention using $n$ to denote the dimension and $\delta$ the distortion bound.  We
now use $N$ to denote the data cardinality.}  Assuming the tolerance parameter $\delta$ fixed, this limitation can
be rephrased as follows: The techniques fail when the number of vectors $N$  is in $\exp\{\Omega(n^{1/2})\}$.

In both Ailon and Chazelle \cite{AilonCh06} and Ailon and Liberty's \cite{AilonLi08} results, as
well as in previous work \cite{JL84}\cite{FranklM87}\cite{DasGuptaGupta99}\cite{Matousek06}\cite{Achlioptas03}
the bounds (\ref{JL}) are obtained by proving strong tail bounds on the distribution of the
estimator $\|\Phi y\|_2$, and then applying a simple union bound
on the finite collection $Y$.  It is worth a moment's thought to realize that Ailon and Chazelle's result
as well as that of Ailon and Liberty can be used for  restricted isometry as well.
Indeed, a simple epsilon-net argument for the set of $r$-sparse vectors can turn that set
into a finite set of $\exp\{O(r\log n)\}$ vectors, on which a union bound can be applied.  However, the current limitation of
random projections mentioned above
will limit $r$  to be in $n^{O(1/2-\mu)}$ (for arbitrarily small $\mu$).
Interestingly, Rudelson and Vershynin's result does not break down for $r$ polynomial in $n$.
A careful inspection of their techniques reveals that instead of union bounding on a finite set
of strongly concentrated random variables, they
use a result due to Dudley to bound extreme values of Gaussian processes.
Can this idea be used
to improve \cite{AilonCh06} and \cite{AilonLi08}?  Intuitively there is no reason why a result which is designed for
preserving the metric of sparse vectors should help with preserving the metric of any finite set of vectors.  It turns
out, luckily, that such a reduction can be done, though not in an immediate way.  A suitable generalization of Rudelson and Vershynin's result
(Section~\ref{rvsec}), combined  with Ailon and Chazelle \cite{AilonCh06} and Ailon and Liberty's \cite{AilonLi08} method of random sign matrix
preconditioning achieves this in Section~\ref{mainresult}.

% It turns out that this is indeed the case.
%This exposition aims at explaining how to combine the two technologies,
%and assumes prior understanding of both \cite{Rudelson06sparsereconstruction} and \cite{AilonLi08}.

\subsection{Notation}\label{notation}
In what follows, we fix $N$ to denote the cardinality of a set $Y$ of vectors in $\R^n$,
where $n$ is fixed.   We also fix a distortion parameter $\delta \in (0, 1/2]$, and
 define $k$ to be an integer in $\Theta(\delta^{-4}(\log N)(\log^4 n))$.

Now let $\Phi$ be a random $k\times n$ matrix obtained by picking $k$ random rows
from the unnormalized $n\times n$ Hadamard matrix (the Euclidean norm of each column of $\Phi$ is $\sqrt k$).
Let $\Omega$ denote the probability space for the choice of $\Phi$.

Let $b$ denote a uniformly chosen vector in $\{-1,1\}^n$, and let  $\Gamma$ denote the probability space on the choice of $b$.
For a vector $y \in \R^n$, we denote by $D_y$ the diagonal $n\times n$ matrix with the coordinates of $y$
on the diagonal. For a real matrix, $\| \cdot \|$ denotes its spectral norm and $(\cdot)^t$ its transpose.
For a set $T \subseteq \{1,\dots n\}$, we let $\id_T$ denote the diagonal matrix with $\id_T(i,i)=1$ if $i\in T$, and $0$ otherwise.
For a vector $y \in \R^n$, let $\supp(y)$ denote the support of $y$, namely, its set of nonzero coordinates.
For a number $p \geq 1 $, let $B_p \subseteq \R^n$ denote the set of vectors $y\in \R^n$
with $\|y\|_p\leq 1$ and  $\alpha B_p$ as the set of vectors $y\in \R^n$ for which $\|y\|_p\leq \alpha$.

%
%%%%% Hessian Math
% \noindent The proof of the following lemma is elementary:
% \begin{lem}\label{hes}
%  Let $f:\R^n\rightarrow \R$ be defined as $f(z) = \|Az\|_2$.  Then the Hessian $H_f$ of $f$ is defined for $z\neq 0$, and is given as
% $$ H_f(z) = \frac {2A^TA}{\|Az\|_2} - \frac{4(A^TAz)(A^TAz)^T}{\|Az\|^3}\ $$
% \end{lem}
%
% \begin{lem}
%  Let $y\in \R^n$ be a unit vector.  Define $g:\R^n\rightarrow R$  as $g(z) = \|\Phi_k^T\Phi D_y z\| - \|D_y z\|$.
% Then ...
% \end{lem}
% \begin{proof}
%
% \end{proof}
%
% \begin{lem}
%  Let $A,B$ be two square matrices with some upper bound $\mu$ on $\|A^TA-B^TB\|$.  Define the function $g:\R^n\rightarrow \R$ as
% $g(z) = \|Az\|_2 - \|(1-\mu)Bz\|_2$.  Then $g$ is globally convex.
% \end{lem}
% \begin{proof}
% By Lemma~\ref{hes}, the Hessian $H_g$ of $g$ is given by $H_1 - H_2$, where
% \begin{eqnarray*}
%  H_1(z) &=& \frac {2A^TA}{\|Az\|_2} - \frac {2(1-\mu)^2B^TB}{\|(1-\mu)Bz\|_2}\ ,  \\
%  H_2(z) &=& \frac{4(A^TAz)(A^TAz)^T}{\|Az\|^3} - \frac{4(1-\mu)^4(B^TBz)(B^TBz)^T}{\|(1-\mu)Bz\|^3} \ .\\
% \end{eqnarray*}
% To show that $H_g(z)$ is positive semidefinite for all $z$, we need to show that for all $w\in \R^n$, $w^TH_g(z)w>0$.
% \begin{eqnarray*}
%  w^TH_1(z)w^T &=& \frac {2w^TA^TAw}{\|Az\|_2} - \frac {2(1-\mu)^2w^TB^TBw}{\|(1-\mu)Bz\|_2}
% \end{eqnarray*}
%
% \end{proof}

%\subsection{A generalization of Rudelson and Vershynin's restricted isometry result}\label{rvsec}
\section{ Restricted isometry result generalization}\label{rvsec}
We follow the main path of Rudelson et al. in \cite{Rudelson06sparsereconstruction}
to prove a more general formulation of their main theorem which is more suitable for us here.
\begin{thm}\label{rvl1}[Derived from Rudelson and Vershynin\cite{Rudelson06sparsereconstruction}]
%Let $r=\sqrt k/\log^4 n$.  Then %, and let $B_r$ denote the set of all vectors $y\in \R^n$ such that $\|y\|_2=1$ and $1\leq \supp(y)\leq r$.
 %For a vector $y\in \R^n$, let $\supp(y)$ denote the support of $y$
%Let $r = k/\log^4n$.  Then
Let $\alpha>0$ be any real number.  Define $E_\alpha$ as
\begin{equation}\label{gkgk}
E_\alpha=E_{\Omega}\left [ \sup_{y\in B_2\cap\alpha B_\infty}{\left \|  D_y^2 - \frac{1}{k}D_y\Phi^{t}\Phi D_y\right \|}\right ]\ .
\end{equation}
%\noindent and let $s(n,k) = \log^{3/2}(n)\log^{1/2}(k)$.
\noindent Then for some global $C_1>0$,
\begin{equation}\label{qweqwe} E_\alpha \le \frac{C_1 \log^{3/2}(n)\log^{1/2}(k)}{\sqrt k}(E_\alpha+\alpha^2)^{1/2}\ .
\end{equation}
In particular, if $\frac{(\log^{3/2}n)(\log^{1/2}k)}{\sqrt k}  = O(\alpha)$, then
\begin{equation}\label{tyutyu} E_\alpha = O\left( \frac{\alpha (\log^{3/2}n)(\log^{1/2}k)}{\sqrt k}\right ) \ .
\end{equation}
%$$ E_\alpha \le \frac{C_1\log^{3/2} n\log^{1/2}k}{\sqrt k}(E_\alpha+\alpha^2)^{1/2}\ .$$
%In particular, if $\frac{C_1\log^{3/2} n\log^{1/2}k}{\sqrt k}  = O(\alpha)$, then
%$$ E_\alpha = O\left(\frac{\alpha\log^{3/2} n\log^{1/2}k}{\sqrt k}\right ) \ .$$
%$$ E_{\Omega}\left [\sup_{|T|\leq r} \left \{{\left \| \frac{ \id_T - \frac{1}{k}\id_T\Phi_{k}^{t}\Phi_k\id_T} {\sqrt {|T|}} \right \|} \right\} \right ] =  O\left(\frac {\log^{3/2}n\log^{1/2}k}{\sqrt k}\right)\ .$$
\end{thm}
%The proof is deferred to Section~\ref{proofrv}

The proof we present is an adaptation of the proof of Theorem~3.4 in \cite{Rudelson06sparsereconstruction} to a more general setting.
In fact, the latter theorem \cite{Rudelson06sparsereconstruction} can be obtained as an easy consequence of theorem~\ref{rvl1}
by replacing $\sup_{y\in B_2\cap \alpha B_\infty}$ in (\ref{gkgk}) by $\sup_{y\in \frac 1 {\sqrt{r}}Y_r}$ where $Y_r\subseteq \R^n$ is defined as
the set of  vectors with at most $r$  coordinates equalling $1$ and the remaining coordinates zero.  Indeed,
$\frac 1 {\sqrt r} Y_r \subseteq B_2 \cap  r^{-1/2} B_\infty$.  We can therefore conclude that for $\alpha = \frac 1 {\sqrt r}$, by definition,
\begin{equation*} E_{\Omega}\left [ \sup_{y\in  \frac 1{\sqrt r} Y_r}{\left \|  D_y^2 - \frac{1}{k}D_y\Phi^{t}\Phi D_y\right \|}\right ] \leq E_\alpha\ .
\end{equation*}

If we also assume that $k = \Theta(r\log^4 n)$, then  (\ref{tyutyu}) will hold, from which we conclude that
\begin{equation}\label{palpal} E_{\Omega}\left [ \sup_{y\in  \frac 1{\sqrt r}Y_r}{\left \|  D_y^2 - \frac{1}{k}D_y\Phi^{t}\Phi D_y\right \|}\right ] \leq O\left(\frac {(\log^{3/2}n)(\log^{1/2}k)}{\sqrt {rk}}\right) .
\end{equation}
Now we notice that $D_y = \frac 1 {\sqrt r} \id_{\supp{y}}$, where for a set of indexes $T$ the diagonal matrix $\id_T$ (as defined in \cite{Rudelson06sparsereconstruction})
has $1$ in diagonal position $i$ if and only if $i\in T$.  Using this observation and multiplying (\ref{palpal}) by $r$ we conclude that
\begin{equation*} E_{\Omega}\left [ \sup_{|T|\leq r}{\left \|  \id_{T} - \frac{1}{k}\id_T\Phi^{t}\Phi \id_T\right \|}\right ] \leq O\left(\frac {\sqrt r (\log^{3/2}n)(\log^{1/2}k)}{\sqrt {k}}\right)\ ,
\end{equation*}
which is exactly the main result of Rudelson and Vershynin in \cite{Rudelson06sparsereconstruction}   for
restricted isometry.

The proof of Theorem~\ref{rvl1} below points out the necessary changes to the proof of Theorem~3.4 in \cite{Rudelson06sparsereconstruction}.
The difference between the theorems is that in our case, the supremum in the definition of $E_\alpha$ is taken not only over the set
of sparse vectors, but over a richer set.  It turns out however that \cite{Rudelson06sparsereconstruction} uses sparsity in a very
limited way:  In fact, the dominating effect of sparsity there is obtained using the fact that the $L_1$ norm of a sparse vector is
small, compared to its $L_2$ norm.   These  arguments appear at the very end of their proof.  For the sake of contributing to the
 self containment of the paper we walk
 through the main milestones of the proof of Theorem~3.4 in \cite{Rudelson06sparsereconstruction},
and point out the changes necessary for our purposes.
The  reader is nevertheless encouraged to refer to the enlightening exposition in
\cite{Rudelson06sparsereconstruction} first.

%Let $r=\Theta(\delta^{-2}\log N)$. 	
%We can assume by our choice of $\alpha$ and $k$ (see Section~\ref{notation}), that $\Phi$ is such that
%\begin{equation}\label{assump}
% \sup_{y\in B_2\cap\alpha B_\infty} \left \|  D_y^2 - \frac{1}{k}D_y\Phi^{t}\Phi D_y\right \| = O\left(\frac{\alpha\log^{3/2} n\log^{1/2}k}{\sqrt k}\right )\ .
%\end{equation}

\begin{proof}
Clearly $E[\frac 1 k D_y\Phi^t\Phi D_y] = D_y^2$.
We define new independent random i.i.d. variables $\epsilon_1,\dots, \epsilon_n$ obtaining
each the values $+1,-1$ with equal probability.  Let $\Pi$ denote the probability space for $\epsilon_1,\dots, \epsilon_n$.
It
suffices to prove  (using  a symmetrization argument, see Lemma~6.3 in \cite{LT91}) that
\begin{equation}\label{gggggg}
E_{\Omega \times \Pi}\left[\sup_{y\in B_2\cap\alpha B_\infty} \left \|\frac 1 k  \sum_{i=1}^k \epsilon_i (x_i D_y)^t(x_i D_y)\right \|
\right] \leq \frac{2C_1(\log^{3/2} n)(\log^{1/2}k)}{\sqrt k}(E_\alpha+\alpha^2)^{1/2} ,\end{equation}
 where $x_i$ is the (random)
$i$'th row of $\Phi_k$.
To that end, as claimed in \cite{Rudelson06sparsereconstruction} (Lemma 3.5),
 if we can show that for any fixed choice of $\Phi$,
\begin{equation}\label{ptptptpt}
E_{\Pi} \left [ \sup_{y\in B_2\cap\alpha B\infty}\left \| \sum_{i=1}^k
\epsilon_i (x_iD_y)^t(x_iD_y)\right \|\right ] \leq k_1
\sup_{y\in B_2 \cap \alpha B_\infty}\left \| \sum_{i=1}^k (x_i D_y)^t(x_iD_y)\right \|^{1/2}
\end{equation}

\noindent for some number $k_1$, then by taking $E_\Omega$ on both sides and using Jensen's inequality (to swap  $(\cdot)^{1/2}$ on the RHS with $E_\Omega$)  and the
triangle inequality, the conclusion
would be that
\begin{equation}\label{tytyty} E_\alpha \leq \frac {2k_1}{\sqrt k}\left (E_\alpha+\|D_y^2\|\right )^{1/2}\ .
\end{equation}
Since $\|D_y^2\| = \|y\|^2_\infty \leq \alpha$, we would get the stated result.  It thus suffices to prove (\ref{ptptptpt})
with $k_1 = O((\log^{3/2} n)(\log^{1/2}k))$.
To do so, \cite{Rudelson06sparsereconstruction} continue by replacing the $k$ binary random variables $\epsilon_1,\dots, \epsilon_k$ in (\ref{ptptptpt})
with $k$ Gaussian random variables $g_1,\dots, g_k$ using a comparison principle (inequality (4.8) in \cite{LT91}), reducing the problem
to that of bounding the expected extreme value of a Gaussian process.  Using Dudley's inequality (Theorem 11.17 in \cite{LT91}), as Rudelson
and Vershynin do, one concludes that (\ref{ptptptpt}) will hold with $k_1$ taken as:
\def\N{\mathcal N}
\begin{equation}\label{intgrl} \int_{0}^\infty \log^{1/2} \N(B, \|\cdot \|_X, u) du\ ,\end{equation}
where:
\begin{itemize}
\item For a norm $\|\cdot\|_\star$, a set $S$ and number $u$, $\N(S, \|\cdot\|_\star, u)$ denotes the minimal number of balls of radius $u$ in norm $\|\cdot\|_\star$ centered in points of $S$ needed to cover the set $S$,
    \item $B$ is defined as $\cup_{y\in B_2 \cap \alpha B_\infty} B_y$, where $B_y= \{D_y z: z\in B_2\}$, and
    \item $\| x\|_X = \max_{i\leq k} |\langle x_i, x\rangle|$, where we remind the reader that $x_i$ is the $i'th$ row of $\Phi$.
\end{itemize}

 Rudelson and Vershynin derive bounds on $\N(B_{RV}, \|\cdot\|_X, u)$ for small $u$ and for large $u$ separately,
where in their case $B_{RV}$ was the set of $r$-sparse vectors of Euclidean norm $1$ (denoted by $D_2^{r,n}$ in \cite{Rudelson06sparsereconstruction}).  The sparsity of the vectors in the set $B_{RV}$ is used in both derivations, as follows:

%Note that by Cauchy-Schwartz, $B \subseteq B_1$.
%This allows us to easily generalize the
%proof of Lemma 3.5 in \cite{Rudelson06sparsereconstruction}, even though that lemma was proven
%for sparse vectors, and there is no reason here for $B$ to
%contain only sparse vectors.

%Indeed, the sparsity there was used
%in two places:
\begin{itemize}
\item For large $u$, they use containment argument (11) in \cite{Rudelson06sparsereconstruction}, asserting
 that $B_{RV}\subseteq \sqrt r B_1$.  Note that by Cauchy Schwartz and the definition of $B$, $B\subseteq B_1$
  hence we "gain" a factor of $\sqrt r$ when deriving $k_1$.
  \item For small $u$,  inequality (13) in \cite{Rudelson06sparsereconstruction} asserts that $\N(B_{RV}, \|\cdot\|_X, u) \leq d(n,r)(1+2/u)^r$, where $d(n,r)$ is the number of ways to choose
$r$ elements from a set of $n$ elements.  Since the best sparseness we can
assume for vectors in $B$ here is trivially $n$, we  replace the expression $d(n,r)$  with $d(n,n)=1$, and $(1+2/u)^r$  with $(1+2/u)^n$.\footnote{To be exact, in \cite{Rudelson06sparsereconstruction} they use the expression $(1+2K/u)^r$ and not $(1+2/u)^r$, but the parameter $K$ in their work can be taken as $1$ for our purposes.}
\end{itemize}
Rudelson and Vershynin then derive a bound for $\int_0^\infty \N^{1/2}(B_{RV}, \|\cdot\|_X, u)du$ by balancing the two bounds at
$u = 1/\sqrt r$.  In our case we balance at $u=1/\sqrt n$.  The net result
will lead to a $k_1$ which is as
the one in the statement of Lemma 3.5 \cite{Rudelson06sparsereconstruction}, except that the
$\sqrt r$ will disappear and $\log r$ will be replaced by
$\log n$.  The conclusion is that we can take $k_1$ to be
$$ k_1 = O\left ((\log n)(\sqrt{\log n})({\log k})\right) = O\left((\log^{3/2}n)(\log k)\right)\ ,$$
as required.
\end{proof}

%The theorem is, in fact, a generalization of Theorem~3.4 in \cite{Rudelson06sparsereconstruction}.
%Before we show how to perform the generalization, we show how to use the Theorem.

%Indeed, this event happens with probability at least $0.99$ assuming the statement of Theorem~\ref{rvl1}.
%In particular, from (\ref{assump}) we can conclude that for all vectors $z\in B_2$ with $|\supp(z)| \leq r$,
%\begin{equation}\label{assump2} 1-O(\delta) \leq \|\Phi z\| \leq 1+O(\delta)\ .
%\end{equation}
% To see this, let $T=\supp(z)$ and take $y$ to be the vector with $y_i=1/\sqrt r$ for all $i\in T$, and $y_i = 0$ otherwise.
%Clearly $y\in B_2 \cap \alpha B_\infty$.  Also, $D_y^2 = \frac 1 r\id_T$.  By (\ref{assump}),
%we see that
%$\left \|\id_T - \frac 1 k \Phi^t\Phi\right \| = O\left(r\alpha(\log^{3/2} n)(\log^{1/2} k)/\sqrt k\right)= O\left(\sqrt r(\log^{3/2} %n)(\log^{1/2} k)/\sqrt k\right)$
%which is $O(\delta)$ for our choice of $r$ and $k$.  This is precisely the guarantee of Theorem~3.4 in \cite{rvl1}.\footnote{There
%the guarantee is written in terms of expectation, here with high probability.}

\section{Random Projections}\label{mainresult}
%To see why Theorem~\ref{rvsec} theorem is a generalization of \cite{Rudelson06sparsereconstruction}
%To see why, let $r$ denote an integer in $\Theta(\delta^{-2}\log N)$ and $\alpha = 1/\sqrt r = \Theta(\delta/\sqrt{\log N})$.

\noindent Our main result claims that the same construction used by Rudelson et al.
also gives improved bounds for random projections.
In what follows, we fix $r$ to be $\lceil \delta^{-2}\log N \rceil$ and $\alpha$ to be $1/\sqrt r$.  Additionally, we
 assume that $\Phi$ is such that
 \begin{equation}\label{assump}
 E_\alpha = O(\alpha^2)\ .
 \end{equation}
   Indeed, Theorem~\ref{rvl1} guarantees that this holds with
 probability at least $0.99$ in $\Omega$.
%The result uses both forms of Theorem~\ref{rvl1}.

\begin{thm}
Let $Y\subseteq B_2$ denote a set of cardinality $N$, and let $\Phi$ satisfy (\ref{assump}).  With probability at least $0.98$ (in $\Gamma$)
we have the following uniform bound for all $y\in Y$:
$$ 1-O(\delta) \leq \left\|\frac 1{\sqrt k} \Phi D_y b\right \| \leq 1 + O(\delta)\ .$$
\end{thm}

We provide  some intuition for the proof.  We split our input vectors $Y$ into sums of two vectors,
one of which is $r$-sparse and the other with $\ell_\infty$ norm bounded by $1/\sqrt{r}$.
We use Rudelson et al.'s original result for the sparse part and
our generalization of it (Theorem~\ref{rvl1}), together with Talagrand's measure concentration theorem
 for the $\ell_\infty$-bounded part.

\begin{proof}
%Let $r$ denote an integer in $\Theta(\delta^{-2}\log N)$.
Let $r$ and $\alpha$ be defined as in Section~\ref{rvsec}.
For each $y\in Y$ we write $y=\yh+\yl$, where $\yh$ is the restriction of $y$ to its $r$ largest (in absolute value) coordinates and
$\yl$ is the restriction to its remaining coordinates.  Note that $\|y\|^2 = \|\yh\|^2 + \|\yl\|^2$ and that $\yh$ is $r$-sparse and that  $\|\yl\|_\infty \leq \alpha$.
$$\left \|\frac{1}{\sqrt{k}}\Phi D_{y} b\right\|^2  = \left \|\frac{1}{\sqrt{k}}\Phi D_{\yh} b\right\|^2  + \left \|\frac{1}{\sqrt{k}}\Phi D_{\yl} b\right\|^2 +
\frac{2}{k}b^{t}D_{\yh}\Phi^{t}\Phi D_{\yl} b. $$
For the first term we have $ \left \|\frac{1}{\sqrt{k}}\Phi D_{\yh} b\right\|^2 = \|\yh\|^2 + O(\delta)$ from Theorem~\ref{rvl1} and
the fact that $\yh$ is $r$-sparse.

In what follows we will use the bound on $\|\yl\|_\infty$
to show that  with high probability, for all $y\in Y$, $\left \|\frac{1}{\sqrt{k}}\Phi D_{\yl} b\right\|^2 = \|\yl\|^2 + O(\delta)$.
A similar argument will bound the cross product $\frac{2}{k}b^{t}D_{\yh}\Phi^{t}\Phi D_{\yl} b$.  %= O(\delta)$.
Combining the three gives the desired result that $\left \|\frac{1}{\sqrt{k}}\Phi D_{y} b\right\|^2 = \|y\|^{2} + O(\delta)$.

We start by analyzing  the measure concentration properties of $ \left \|\frac{1}{\sqrt{k}}\Phi D_{\yl} b\right\|^2$.  % considering
%only the probability space $\Gamma$.
Let $X_{\yl}$ be the Rademacher random variable defined by
$$ X_{\yl} = \left \|\frac 1 {\sqrt k}\Phi D_{\yl} b\right\|\ .$$
Let $\mu_\yl$ denote a median of $X_\yl$.
By Talagrand \cite{LT91}, we have that for all $t>0$,
\begin{eqnarray}
 \Pr[X_\yl > \mu_\yl + t] &\leq& \exp\{-C_2 t^2/\sigma_\yl^2\}  \label{hththt1}\\
\Pr[X_\yl < \mu_\yl - t] &\leq& \exp\{-C_2 t^2 / \sigma_\yl^2\}\label{hththt2}
\end{eqnarray}
for some global $C_2$, where $\sigma_\yl = \left \|\frac 1 {\sqrt k}\Phi D_{\yl}\right \|$.
By the triangle inequality and Equation~(\ref{assump}) we have
$\sigma_\yl^2 = \| \frac{1}{k}D_\yl \Phi^{t}\Phi D_\yl - D_\yl^2 + D_\yl^2 \|  \leq  \alpha^2  + \|D_\yl^2 \|$.
Clearly $\|D_\yl\| = \| \yl \|_\infty \leq \alpha$.  Hence, $\sigma_\yl^2 = O(\alpha^2)$.
%We now assert that $\mu_\yl$ is close in value to $\|\yl\|$.
%It is easy to verify, on the one hand, that $E[X_\yl^2] = \|\yl\|^2$.
%On the other hand,
%\begin{equation}
%\mu_\yl^2 +  \sum_{t=0}^\infty   (t^2- 2\mu_\yl t)\exp\{-C_2t^2/\sigma_\yl^2\} \leq E[X^2] \leq \mu_\yl^2 + \sum_{t=0}^\infty (t^2+2\mu_\yl t)\exp\{-C_2t^2/\sigma_\yl^2\}\ .
%\end{equation}
%Assuming  $\alpha$ satisfies $\alpha \geq exp\{-C_2/\alpha^2\}$ (this happens for $N$ sufficiently large - a harmless assumption),
%On the other hand, from Appendix~\ref{appendix1} and (\ref{hththt1})-(\ref{hththt2}) we have that $E[X_\yl^2] = \mu_\yl^2 \pm O(\sigma_\yl^2 + \sigma_\yl \mu_\yl)$.
%We conclude that $\|\yl\| - O(\sigma_{\yl}) \leq \mu_\yl \leq \|\yl\| + O(\sigma_\yl)$.
From the fact that $E[X_\yl^2] = \|\yl\|^2$ and using Appendix~\ref{appendix1} and (\ref{hththt1})-(\ref{hththt2})
We conclude that $\|\yl\| - O(\sigma_{\yl}) \leq \mu_\yl \leq \|\yl\| + O(\sigma_\yl)$.
Hence, again using (\ref{hththt1})-(\ref{hththt2}) and union bounding over the $N$ vectors in $Y$,
we conclude that with probability $0.99$, uniformly for all $y \in Y$:
$$ \|\yl\| - O(\delta) \leq \frac 1 {\sqrt k}\left \| \Phi D_{\yl} b\right \| \leq \|\yl\| + O(\delta)\ .$$

%By \ref{assump2} we know that the following bound holds for all $\yh$ and for all $b$:
%$$ \|\yh\| - O(\delta)\|\yh\| \leq \frac 1 {\sqrt k}\left \| \Phi D_{\yh} b\right \| \leq \|\yh\| + O(\delta)\|\yh\|\ .$$

We now bound the cross term $Z =  \frac 1 k b^t D_{\yh} \Phi^t\Phi D_{\yl} b$  ($y$ is now held fixed).
By disjointness of $\supp(\yh)$ and $\supp(\yh)$, $E[Z] =0$. Decompose $b$ into $\bl + \bh$, where $\supp(\bl) = \supp(\yl)$
and $\supp(\bh) = \supp(\yh)$.  For any fixed $\bh$, the function $Z$ is linear (and hence convex) in $\bl$.  Also for all possible
values $\bh'$ of $\bh$, $E[Z| \bh=\bh'] = 0$.  Hence, again by Talagrand,
\begin{eqnarray}\label{hththt}
 \Pr[Z > \mu_{\bh'} + t] &\leq& \exp\{-C_2 t^2/\sigma_{\bh'}^2\} \\
\Pr[Z < \mu_{\bh'} - t] &\leq& \exp\{-C_2 t^2 / \sigma_{\bh'}^2\}
\end{eqnarray}
where $\mu_\bh'$ is a median of $(Z|\bh = \bh')$, and $\sigma_{\bh'}= \|\frac 1 k(\bh')^t D_{\yh}\Phi^t \Phi D_{\yl}\|$.  Clearly,
\begin{eqnarray*}
\sigma_{\bh'} &\leq& \left \|\frac 1 {\sqrt k}(\bh')^t D_{\yh}\Phi^t \right \| \cdot \left\|\frac 1 {\sqrt k}\Phi D_{\yl}\right\| \\
&=& O(\|\yh\|\sigma_\yl) = O(\sigma_\yl) = O(\alpha)\ .
\end{eqnarray*}
Again using Appendix~\ref{appendix1} and $E[Z| \bh=\bh'] = 0$ gives
that $|\mu_\bh'| = O(\alpha)$, and again we conclude using a union bound that with probability at least
$0.99$, uniformly for all $y\in Y$,
$\left |\frac 1 k b^t D_{\yh} \Phi^t\Phi D_{\yl} b\right | = O(\delta)$.

Tying it all together, we conclude that with probability at least $0.98$, uniformly for all $y\in Y$,
\begin{eqnarray*}
\frac 1 k\| \Phi D_yb\|^2 &=& \frac 1 k \|\Phi D_{\yl}b\|^2 + \frac 1 k \|\Phi D_{\yh}b\|^2 + 2b^tD_{y^H}\Phi^t\Phi D_{\yl} b \\
&=& \| y \|^2 + O(\delta)\ , \\
\end{eqnarray*}
as required.
\end{proof}

\section{Conclusions}
The obvious problems left open are those  of (1) improving the dependence of $k$ in $\delta$ (from $\delta^{-4}$ to $\delta^{-2}$) and (2) removing the dependence of $k$ in $\polylog(n)$.  Other directions of research include not only reducing the computational efficiency of random dimension reduction, but also the amount of randomness needed for the construction.

\section*{Acknowledgements}
We thank Emmanuel Candes for helpful discussions.

\bibliographystyle{unsrt}
\bibliography{fjlt}
%\begin{thebibliography}{10}
%\bibitem{rv} the Vershynin and Rudelson paper

%\bibitem{AC} The Ailon-Chazelle paper

%\bibitem{AL} The Ailon-Liberty paper
%\bibitem{JL} The Johnson-Lindenstrauss paper

%\bibitem{JLS} Versions of Johnson-Lindenstrauss (e.g. Achlioptas...)

%\bibitem{LT} Ledoux-Talagrand book

%\end{thebibliography}

\appendix
\section{}
\label{appendix1}
\begin{fact}
For any real valued random variable $Z$ such that for all $t>0$
\begin{eqnarray}
\label{lalala} \Pr[Z > \mu + t] &\leq& \exp\{ - c t^2/\sigma^2\} \\
\nonumber \Pr[Z < \mu - t] &\leq& \exp\{ - c t^2 / \sigma^2\}
\end{eqnarray}
we have that $\sqrt{E(Z^2)} - O(\sigma) \le \mu \le \sqrt{E(Z^2)} + O(\sigma)$.
\end{fact}
\begin{proof}
Define the variable $Z' = (Z - \mu)/\sigma$.
\begin{eqnarray*}
E[Z'] \le E[|Z'|] &\le& \sum_{i=1}^{\infty} i \Pr(i-1\le |Z'| \le i) \\
&\le& \sum_{i=1}^{\infty} i \Pr(|Z'| \ge i - 1) \le 2\sum_{i=1}^{\infty} i \exp\{-c (i-1)^2\} = O(1)
\end{eqnarray*}
Clearly, $E[Z'] = O(1)$ gives $E(Z) = \mu + O(\sigma)$. In the same way we get $E[Z'^2] = O(1)$.
Thus, $E[Z^2] - 2\mu E[Z] + \mu^2 = O(\sigma^2)$ and $E[Z^2] = (\mu \pm O(\sigma))^2$

%Let $\psi$ be the probability distribution of the variable $Z$, then
%$E(Z^2) = \int_{-\infty}^{\infty} z^2 \psi(z) dz$.
%Notice that $ \psi(z) dz = d \Pr(Z \le z) = -d \Pr(Z \ge z)$. Splitting the
%integration range and integrating by parts we get
%\begin{eqnarray}
%\int_{-\infty}^{\infty} z^2 \psi(z) dz &=& z^{2}\Pr(Z \le z) |_{-\infty}^{\mu} - 2\int_{-\infty}^{\mu}z\Pr(Z \le z) dz \\
%\nonumber &&-  z^{2}\Pr(Z \ge z) |_{\mu}^{\infty} + 2\int_{\mu}^{\infty}z\Pr(Z \ge z) dz \\
%\nonumber &=& \mu^2 + O(\sigma^2 + \sigma \mu).
%\end{eqnarray}
%The last step is achieved by substituting $z$ by $\mu-t$ or $\mu+t$ and integrating
%over the bounds on $ \Pr[Z > \mu + t]$ and $\Pr[Z < \mu - t]$ from Equation~(\ref{lalala}).
\end{proof}

\end{document}